\documentclass[sigconf,nonacm,10pt,margin=1in]{acmart}
\AtBeginDocument{%
  }

\usepackage[linesnumbered, ruled]{algorithm2e}
\usepackage[noend]{algpseudocode}

\SetCommentSty{mycommfont}
\SetKwInput{KwInput}{Input}                
\SetKwInput{KwOutput}{Output}              

\usepackage{booktabs}
\usepackage{multirow}

\usepackage{graphicx}
\usepackage{subcaption}
\usepackage{stfloats}

\usepackage{adjustbox}

\usepackage{tikz}
\usepackage[framemethod=TikZ]{mdframed}
\mdfdefinestyle{MyFrame}{%
	linecolor=black,
	outerlinewidth=1pt,
	roundcorner=0pt,
	innerrightmargin=5pt,
	innerleftmargin=5pt,
	backgroundcolor=white
}
\usepackage[skip=8pt]{caption}
\usepackage{thm-restate}

\newtheorem{claim}{Claim}
\newtheorem{definition}{Definition}

\newcounter{example}

\usepackage{xspace}
\newcommand{\oursystem}{\textsf{PostRI}\xspace}
\newcommand{\DomainSize}{\ensuremath{N}}
\newcommand{\Domain}{\ensuremath{\mathcal{Y}}}
\newcommand{\StepSize}{\ensuremath{s}}
\newcommand{\MaxUtility}{\ensuremath{y^*(D)}}
\newcommand{\MaxRIUtility}{\ensuremath{b^*(D)}}
\newcommand{\Dataset}{\ensuremath{D}}
\newcommand{\DataSize}{\ensuremath{n}}
\newcommand{\EMBound}{\ensuremath{\gamma}}

\newcommand{\MedianSensitiv}{\ensuremath{\Delta_u}}
\newcommand{\MUtilSymbol}{\ensuremath{u}}
\newcommand{\MedianUtility}[1]{\ensuremath{\MUtilSymbol(D, #1)}}
\newcommand{\RIUtilSymbol}{\ensuremath{q}}
\newcommand{\RIUtility}[1]{\ensuremath{\RIUtilSymbol(D, #1)}}
\newcommand{\RISensitiv}{\ensuremath{\Delta_q}}
\newcommand{\Rank}[1]{\ensuremath{\text{R}(D, #1)}}

\newcommand{\TrueMedian}{\ensuremath{m}}
\newcommand{\NoisyMedian}{\ensuremath{o}}
\newcommand{\RIStep}{\ensuremath{b}}
\newcommand{\Lipshitz}{\ensuremath{\ell}}
\newcommand{\Helper}{\ensuremath{f}}
\newcommand{\HelperFunction}[1]{\ensuremath{\Helper_{#1}(\Dataset)}}
\newcommand{\ExpMech}[2]{\ensuremath{Exp_{#2}(\Dataset,{#1})}}

\newcommand{\mc}[1]{\mathcal{#1}}

\usepackage{amsmath}

\newcommand{\eat}[1]{}

\newcommand{\squishlist}{
	\begin{list}{$\bullet$}
		{
			\setlength{\itemsep}{0pt}
			\setlength{\parsep}{3pt}
			\setlength{\topsep}{3pt}
			\setlength{\partopsep}{0pt}
			\setlength{\leftmargin}{1.5em}
			\setlength{\labelwidth}{1em}
			\setlength{\labelsep}{0.5em} } }

\newcommand{\squishend}{
	\end{list}  }

\usepackage{enumitem}

\usepackage{mathtools}

\begin{document}

\title{Interpreting the Error of Differentially Private Median Queries through Randomization Intervals}

\author{\large Thomas Humphries\textsuperscript{1}, Tim Li\textsuperscript{1}, Shufan Zhang\textsuperscript{1}, Karl Knopf\textsuperscript{1}, Xi He\textsuperscript{1,2}}
\affiliation{%
  \institution{\textsuperscript{1}University of Waterloo, \textsuperscript{2}Vector Institute}
  \country{\textit{\{thomas.humphries, tlli, shufan.zhang, karl.knopf, xi.he\} @uwaterloo.ca}}
}

\begin{abstract}

It can be difficult for practitioners to interpret the quality of differentially private (DP) statistics due to the added noise.
One method to help analysts understand the amount of error introduced by DP is to return a Randomization Interval (RI), along with the statistic.
A RI is a type of confidence interval that bounds the error introduced by DP.
For queries where the noise distribution depends on the input, such as the median, prior work degrades the quality of the median itself to obtain a high-quality RI.
In this work, we propose \oursystem, a solution to compute a RI after the median has been estimated.
\oursystem enables a median estimation with 14\%-850\% higher utility than related work, while maintaining a narrow RI.

\end{abstract}

\maketitle
\pagestyle{plain}
\section{Introduction}
A core functionality of database management software is allowing analysts to query the database to learn aggregate statistics.
Naively allowing analysts to obtain unperturbed statistics can open up a lucrative attack surface over sensitive data~\cite{DinurN03}.
A popular defence to such attacks is Differential Privacy (DP)~\cite{Dwork06differentialprivacy}.
By adding a calibrated noise to each query result, the data owner enjoys a formal guarantee of privacy over the result of the queries.
A challenge with DP is that it requires a certain level of expertise to apply correctly, as well as to interpret the results of a DP query.

In this work, we focus on the problem of interpreting the level of error introduced by DP to satisfy privacy.
One method to help analysts understand this error is to return a confidence interval alongside the noisy result.
Specifically, an upper and lower bound on the noisy statistic such that the true answer is contained in the interval with high confidence (e.g., $95\%$).
To avoid confusion with confidence intervals from statistics that bound sampling error, we introduce the term \emph{randomization intervals (RI)} for a confidence interval that bound only the randomness of the DP mechanism itself.
These randomization intervals allow the analyst to reason about the error of the query without understanding the DP mechanism.
For example, if the interval is small, the analyst can be confident that the result is accurate.
Conversely, consider an analyst conducting a count query to determine the number of people in the dataset with a specific attribute. 
If the RI for this query contains zero, the user may not trust that there are any people in the dataset with this attribute.

For queries such as sums and counts, where the randomness is independent of the input, one can use a tail bound to compute the RI at no additional privacy cost.
However, this method does not extend to more general queries, such as the median, that require a data-dependent DP noise distribution (e.g. the exponential mechanism) for good utility.
Sun et al. were the first to study the problem of randomization intervals for the median (and other queries)~\cite{sun2023confidence}.
Instead of returning a high utility estimation for the median along with an estimated RI, Sun et al. focus on estimating the RI first, and then post-process the median as the average of the RI bounds.
The challenge with this approach is that the utility of the median becomes significantly degraded as estimating an RI has an inherently larger error factor to ensure high confidence.
Furthermore, the private data is not guaranteed to be evenly distributed within the RI, making the center of the RI a poor estimation of the median.

We introduce \oursystem, a new approach to computing the RI alongside the median, while maintaining the high utility of the classic DP median computation.
\oursystem first computes the DP median using an unmodified exponential mechanism approach. 
Then, using our single-shot, low-sensitivity utility function, we estimate the RI using a tunable amount of privacy budget left over from the median.
\oursystem outputs an RI of similar width to Sun et al.~\cite{sun2023confidence}, while significantly improving the error of the median query itself.
We prove the privacy and correctness of \oursystem and derive optimal values for hyperparameters such as the ratio of privacy budget between the median and RI.
We evaluate \oursystem over real-world datasets and find a 14\%-850\% improvement in the average median error over~\citet{sun2023confidence}, while maintaining approximately the same or a slightly larger RI width, depending on the dataset.

\section{Background}
Differential privacy (DP)~\cite{Dwork06differentialprivacy} guarantees that an algorithm's output is approximately the same, regardless of the participation of any single user.
More formally, differential privacy can be defined as follows.
\begin{definition}[Differential Privacy (DP)\label{def:dp}]
A randomized algorithm $M: \mc{D} \rightarrow  \mc{O}$ 
is $\epsilon$-DP if for any pair of neighboring databases $\Dataset, \Dataset' \in \mc{D}$, and all $O \subseteq \mc{O}$ we have $\Pr[M(D) \in O] \leq e^\epsilon \Pr[M(D') \in O].$
\end{definition}
We use the bounded neighbouring definition where datasets are neighbours if they differ in the replacement of a single record, $|\Dataset \cap \Dataset'| = \DataSize - 1$. 
Note that if we apply a DP mechanism(s) sequentially, the privacy parameters are composed through summation or more advanced methods~\cite{dptextbook14}.

The exponential mechanism, introduced by McSherry and Talwar, is a general-purpose DP mechanism that maximizes a given utility function privately~\cite{mcsherry_2007}. 
\begin{definition}[Exponential Mechanism~\cite{mcsherry_2007}]\label{defn:exp_mech}
Given privacy budget $\epsilon$ and utility function $u: \mc{D} \times \Domain \to \mathbb{R}$ with sensitivity $\MedianSensitiv:= \max_{\Dataset, \Dataset'\in \mc{D}, y \in \Domain}\left|u(\Dataset, y) - u(\Dataset', y)\right|,$
the \emph{exponential mechanism} $\ExpMech{\epsilon}{u}$ outputs a sample $y\in\Domain$, with the following probability $Pr[y] = \exp\left( \frac{\epsilon u(\Dataset, y)}{2\MedianSensitiv}\right)/\sum \limits_{i\in\Domain} \exp\left( \frac{\epsilon u(\Dataset, i)}{2\MedianSensitiv}\right)$.
\end{definition}
The exponential mechanism, as defined above, guarantees $\epsilon$-differential privacy~\cite[Theorem 6]{mcsherry_2007}. 


\subsection{Problem Setup}
\begin{definition}[Randomization Interval]
Given a dataset $\Dataset$ of size $\DataSize$, privacy budget $\epsilon$, and a failure probability $\beta$, we wish to output a  triple $(l, \NoisyMedian, u)$ called a \emph{Randomization Interval (RI)} such that:
\begin{itemize}
    \item outputting $(l, \NoisyMedian, u)$ satisfies $\epsilon$-DP.
    \item $\NoisyMedian$ is a DP estimate of the true median $\TrueMedian$ of dataset $\Dataset$.
    \item With probability $1-\beta$, $\TrueMedian\in[l,u]$.
\end{itemize} 
\end{definition}
Intuitively, we wish to give a lower and upper bound on the noisy median estimate to make the error of the noisy median more interpretable. 
Another way to state this goal is that we are computing a confidence interval with respect to the error introduced by differential privacy (we clarify the relation to related work on DP confidence intervals in Appendix~\ref{sec:related_work}).
We consider computing the true median to be deterministic in this work.

For simplicity, we assume $\Domain$ is an integer domain of size $\DomainSize$, denoted $\Domain = [\DomainSize]=\{0,1,\dots,\DomainSize\}$.
We also assume the utility function $u$ is $1$-Lipshitz,  that is, $|\MedianUtility{y} - \MedianUtility{y+1}|\leq \Lipshitz=1$ for any $\Dataset\in \Domain^\DataSize, y\in \Domain$. In the case of the median, this implies that the dataset $\Dataset$ contains no repeated elements. For datasets that do not satisfy this assumption, we follow Sun et al.~\cite{sun2023confidence} and remap the data to an expanded domain of size $\DomainSize\cdot \DataSize$.
The de-duplicated domain can be found by mapping $\Dataset\in\left[\DomainSize\right]^\DataSize$ to $\tilde{\Dataset} \in \left[\DataSize \cdot \DomainSize \right]^\DataSize $ , where the $k$ repetitions of an element $x \in \Dataset$ are mapped to consecutive elements of the new domain: $\DataSize x, \DataSize x+1, \dots, \DataSize x+k-1$.

\section{\oursystem}
In this section, we first overview the design of \oursystem in the two main steps, the median estimation and then the randomization interval.
We then analyze the privacy and correctness of the complete algorithm and finally discuss the optimal hyperparameters.

\subsection{Median Mechanism}
The first step in our algorithm is to derive a differentially private median using a standard approach~\cite{li_DP_book}.
We give the details of this algorithm for completeness.
Since \oursystem operates in two parts, we must split the privacy budget between this median computation and the randomization interval.
We denote the privacy budget of the median as $\epsilon_1$ and the privacy budget of the RI as $\epsilon_2$, where $\epsilon = \epsilon_1+\epsilon_2$ and discuss how to set these parameters in Section~\ref{sec:hyperparams}.

The utility function we will use for the median is 
\begin{equation}\label{eq:median_utility}
    \MedianUtility{y} = -|\Rank{y} - \frac{|\Dataset|}{2}|
\end{equation}
where $\Rank{y}$ computes the number of data points less than or equal to $y$ after sorting the dataset $D$.

The utility function $\MedianUtility{y}$, has sensitivity of $1$, and thus applying the exponential mechanism satisfies $\epsilon_1$-DP ~\cite[Theorem 6]{mcsherry_2007}.
In the case where multiple values in the domain $\Domain$ have the same utility, it is common to weight the probability of selection by the number of elements with this utility, and then randomly select a candidate after~\cite{joint_exponential}. We do this for the median since there can be many domain values between each data point, all with the same rank.

\subsubsection{Median Utility}\label{sec:median_util}
There is a well-established utility bound on the exponential mechanism that we can apply to get a utility bound on the median~\cite{dptextbook14}.
\begin{theorem}\label{thm:median_utility}
    Let $\MaxUtility = argmax_{y \in \Domain} (\MedianUtility{y})$. 
Then, with probability $1-\beta_1$ the following statement holds~\cite{dptextbook14}:
\begin{equation}
    \MedianUtility{\NoisyMedian} \geq \MedianUtility{\MaxUtility}-\EMBound_1
\end{equation} where 
\begin{equation}\label{eqn:bound_em_1}
    \EMBound_1 = \frac{2\MedianSensitiv}{\epsilon_1} \log{\frac{\DomainSize}{\beta_1}}
\end{equation}
\end{theorem}

\subsection{RI Mechanism Design}

Since there exists a bound on the error of the exponential mechanism (Theorem~\ref{thm:median_utility}), constructing a randomization interval for the median seems simple.
However, this error bound is in terms of the utility function, which presents several problems. 
First, we do not have the utility value of the outputted median as publishing the utility requires additional privacy budget.
Second, even if we publish the utility, the value is relative and is not useful without knowledge of the private dataset.
In the case of the median, a randomization interval using the utility bound would only tell us the rank error.
For example, we could bound the median in a range of up to ten dataset positions, but depending on the dataset, the closest ten data points could be numerically very far away from the median.
Thus, we want to give a randomization interval in the data domain.

A strawman solution could be to compute the utility bound and then map it back onto the data domain.
However, this would use a significant amount of privacy budget, paying both to compute the utility bound and select the values from the data domain.
We instead develop a novel utility function to output the randomization interval in a single application of the exponential mechanism.

In order to save additional privacy budget, we first reformulate the randomization interval problem as the private selection of a single value $\RIStep$ (rather than a separate upper and lower bound).
\begin{mdframed}[style=MyFrame]
    \textbf{Simplified Randomization Interval Problem:} 
    Given an output median $\NoisyMedian$, find a minimal $\RIStep$ such that with probability $1-\beta_2$, while satisfying $\epsilon_2$-DP.
    \begin{center}
        $\TrueMedian \in [\NoisyMedian-\RIStep, \NoisyMedian+\RIStep]$
    \end{center}
    where $\TrueMedian$ is the ground truth median.

\end{mdframed}
Using this simplified formulation, we can then design a utility function for the optimal $\RIStep$. Our first step is to create a helper function $\Helper$ that measures the worst-case coverage of a given $\RIStep$ in terms of rank using dataset $\Dataset$. We define this function as $\HelperFunction{\RIStep} = \min (| \Rank{\NoisyMedian+\RIStep} - \Rank{\NoisyMedian} |, | \Rank{\NoisyMedian} - \Rank{\NoisyMedian-\RIStep}|)$.
Intuitively, $\Helper$ computes the smallest rank distance from the interval boundary to the median.
We note that we chose the above function for two reasons.
First, it allows us to consider both sides of the interval in one shot.
Second, it has sensitivity one (shown in Appendix~\ref{sec:proof_RI_Privacy}), making it very efficient to compute privately.

Recall the goal is to construct a valid randomization interval by making $\RIStep$ large enough that the true median $\TrueMedian$ is contained in the interval with high probability.
To determine this width, we can use the utility bound in Theorem~\ref{thm:median_utility}.
Since we designed the helper function to measure rank distance, and the utility function of the median measures rank distance, we can simply find the minimum $\RIStep$ such that $\HelperFunction{\RIStep} > \EMBound_1$.
A natural way to privately select such a $\RIStep$ would be to use the sparse vector technique (SVT)~\cite{dptextbook14}.
SVT sequentially evaluates a given set of queries and halts (and outputs the query index) when the first query value exceeds this threshold.
In our case, the set of queries would be defined by evaluating the helper function $\Helper$ over a set of $\RIStep$ values and the threshold would be $\EMBound_1$.
However, it was shown by Lyu et al.~\cite{lyu_2017} that in a non-interactive setting such as this (the queries are known ahead of time), replacing SVT with the exponential mechanism gives better utility\footnote{Lyu et al. consider top-$k$ queries, but we find the same result holds for our threshold queries}.
Thus, in our work, we apply the exponential mechanism using a utility function that is the absolute difference between the query set and the threshold.

We note that in addition to the threshold $\EMBound_1$, we must account for two additional sources of error.
The first is due to the fact that we are selecting $\RIStep$ using another exponential mechanism.
We must ensure correctness under the worst-case error of this second exponential mechanism in selecting the error of the first.
Thus, we must incorporate a second error bound similar to Theorem~\ref{thm:median_utility} but with the following constant.
\begin{equation}\label{eqn:bound_em_2}
    \EMBound_2 = \frac{2\RISensitiv}{\epsilon_2} \log{\frac{\DomainSize}{\StepSize\beta_2}}
\end{equation} 
where $\StepSize$ is a quantization parameter determining the domain set we choose for $\RIStep$.
This quantization is the second source of error we must account for.
Specifically, if we consider all possible $\RIStep$ values, we have a larger domain that negatively affects error; if we pick a smaller domain, we must account for the quantization error and the new domain size.
We define the domain set to be $\RIStep \in \{\StepSize, 2\StepSize, \dots, \lfloor \frac{\StepSize \cdot \DomainSize}{\StepSize} \rfloor \}$.
This gives a domain size of $\DomainSize/\StepSize$ and introduces a quantization error.
Namely, $\RIStep$ can be at most $\StepSize$ away from an optimal value. 
Since $\RIStep$ is in the data domain and the utility is in the rank domain, we multiply $\StepSize$ by $\Lipshitz$ (the Lipshitz coefficient assumed to be 1 if no repeated data elements) to obtain a bound on this error in terms of utility. 

Putting all of this together, we get our final utility function.
\begin{equation}\label{eqn:em_ci_utility}
    \RIUtility{\RIStep} = -|\HelperFunction{\RIStep} - \EMBound_1 -\EMBound_2 - \StepSize\cdot\Lipshitz|
\end{equation}
We apply the exponential mechanism on $\RIUtilSymbol$ to obtain an $\epsilon_2$-DP RI mechanism.
\begin{restatable}{theorem}{RIPrivacy}\label{thm:RIPrivacy}
    $\RIUtilSymbol$ has a sensitivity of $1$ and thus applying the exponential mechanism satisfies $\epsilon_2$-DP.
\end{restatable}
We defer the proof of Theorem~\ref{thm:RIPrivacy} to Appendix~\ref{sec:proof_RI_Privacy}. In Appendix~\ref{sec:correctProof}, we prove the correctness of \oursystem. Namely:
\begin{restatable}{theorem}{RICorrect}\label{thm:RICorrect}
    For a given private median estimation $\NoisyMedian$ and randomization interval $[\NoisyMedian-\hat{\RIStep},\NoisyMedian+\hat{\RIStep}]$ output using \oursystem, with probability $1-\beta_1-\beta_2$, we have $\Rank{\NoisyMedian - \hat{\RIStep}}\leq n/2\leq \Rank{\NoisyMedian + \hat{\RIStep}}$
\end{restatable}

\subsection{Hyperparameter Selection}\label{sec:hyperparams}
\oursystem has two hyperparameters that can be varied. The first is the split between $\epsilon_1$ and $\epsilon_2$. That is, more budget can be allotted to reduce the error of the median or to shorten the length of the RI. However, we note that one can not arbitrarily shorten the RI as it inherently depends on the error of the median in our approach. We derive the optimal split between $\epsilon_1$ and $\epsilon_2$ to give the shortest possible RI width in Appendix~\ref{sec:epsilonsplitting}. The result is
\begin{equation}\label{eqn:optimalEps}
    \epsilon_1 = \epsilon_2\sqrt{\frac{\log{\frac{\DomainSize}{\beta_1}}}{\log{\frac{\DomainSize}{\StepSize\beta_2}}}}
\end{equation}
The second parameter is the step size $\StepSize$, which determines the domain of $\RIStep$.
In Appendix~\ref{sec:optimalstepsize}, we derive the optimal parameter setting for $\StepSize$ such that the RI length is minimized:
\begin{equation}\label{eqn:optimalStep}
    \StepSize = \frac{2\RISensitiv}{\epsilon_2 \ell}.
\end{equation}
We note a circular dependence between the optimal choice of $\epsilon_1$, $\epsilon_2$, and $\StepSize$. In practice, we find that substituting these equations into each other iteratively converges to a stable value after a few iterations.

\section{Preliminary Experimental Results}






\begin{table}[t]
    \centering
    \caption{Comparison of Median RI on Real Datasets with $\epsilon = 1$ and $\beta = 0.01$.}
    \label{tab:comparison}
    \begin{adjustbox}{width=\columnwidth}
    \begin{tabular}{|c|c|c|c|c|}
    \hline
    Dataset 
    & Technique
    & Median Error
    & Average RI Width
    & Correctness \\
    \hline \hline 
    \multirow{2}{*}{Bank} & Ours & \textbf{0.06 ($\pm$ 0.24)} & 14.19 ($\pm$ 0.67) & 1.00 \\ \cline{2-5}
    & \cite{sun2023confidence}  & 0.57 ($\pm$ 0.62) & 13.63 ($\pm$ 0.55) & 1.00 \\ \hline
    \multirow{2}{*}{Adult} & Ours & \textbf{32.40 ($\pm$ 28.61)} & 1264.00 ($\pm$ 74.33) & 1.00 \\ \cline{2-5}
    & \cite{sun2023confidence}  & 166.88 ($\pm$ 17.23) & 1146.56 ($\pm$ 17.11) & 1.00 \\ \hline
    \multirow{2}{*}{Airplane} & Ours & \textbf{7.88 ($\pm$ 4.81)} & 13.13 ($\pm$ 2.58) & 1.00 \\ \cline{2-5}
    & \cite{sun2023confidence}  & 9.00 ($\pm$ 0.00) & 9.00 ($\pm$ 0.00) & 1.00 \\ \hline
    \end{tabular}
    \end{adjustbox}
\end{table}

To evaluate our method \oursystem, we conduct a study over three real-world datasets and compare with the existing approach by Sun et al.~\cite{sun2023confidence}.
We give implementation details in Appendix~\ref{sec:implementation}. We focus on the ``balance'' attribute (with values in range $[-8019, 102127]$ and true median value of 448) for the Banking dataset, the ``fnlwgt'' attribute (with values in range $[12285,1490400]$ and true median value of 178144.5) for the Adult dataset, and the ``capacity'' attribute (with values in range $[4,396]$ and true median value of 162) for the Airplane dataset in the evaluation.
We fix $\StepSize$ following Eqn~\ref{eqn:optimalStep}.
We use the average median error (i.e., numerical distance from true median), the average RI width (i.e., the distance between the lower and upper bounds), and the observed correctness (i.e., the percentage of how many times the true median is inside the reported RI) over 100 runs as evaluation metrics.
Table~\ref{tab:comparison} shows the comparison results with privacy budget $\epsilon=1$ and $\beta=0.01$.
Our \oursystem method (default setting, with $\epsilon_1 = \epsilon_2=\frac{1}{2}\epsilon$)  improves the average median error over~\citet{sun2023confidence} by 14\%-850\% with a moderately wider RI width by 3\%-35\%. 
The RI results reported by both methods contain the true median 100\% of the time in our experiments, implying an observed correctness rate of 1.

We also measure the average median error and RI width with varying privacy budgets ($\epsilon=0.25, 0.5, 1, 2, 4$) across datasets. 
We compare \citet{sun2023confidence} and \oursystem with different budget splitting: default ($\epsilon_1 = \epsilon_2$), optimal (cf. Eq.~\ref{eqn:optimalEps}), and median-focused ($\epsilon_1 = 9\epsilon_2$). 
As shown in Fig.~\ref{fig:varying_eps_banking}-\ref{fig:varying_eps_airplane}, our \oursystem has lower median error over all privacy regimes on these datasets and yields a comparable RI width as the method in \citet{sun2023confidence}, which matches our theoretical analysis (Appendix~\ref{sec:utility_analysis}).
The shaded areas depict the standard deviation across 100 runs.
Notably, an advantage of our approach is being able to tune the privacy budget split in favour of the median.
Experiments on the Banking dataset (Fig.~\ref{fig:varying_eps_banking}) show that the median-focused split (i.e., 90\% of the budget is spent on the median) can achieve the most accurate result on high privacy regimes, albeit with a wider RI width, compared to the default setting and the optimal split that minimizes the overall error.

\begin{figure}[tbp]
  \centering
  \begin{subfigure}[t]{0.48\columnwidth}
    \centering
    \includegraphics[width=\linewidth]{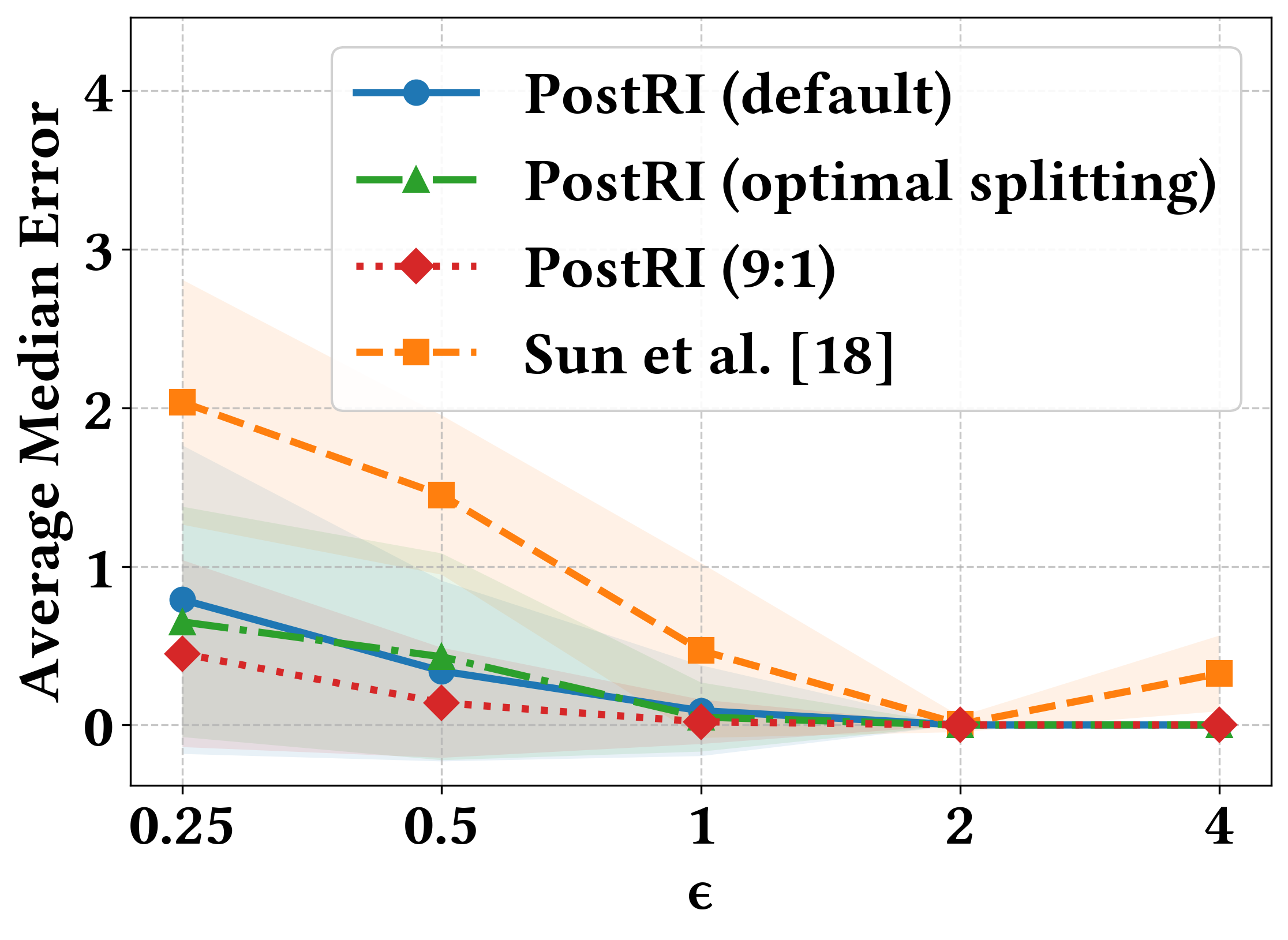}
  \end{subfigure}\hfill
  \begin{subfigure}[t]{0.48\columnwidth}
    \centering
    \includegraphics[width=\linewidth]{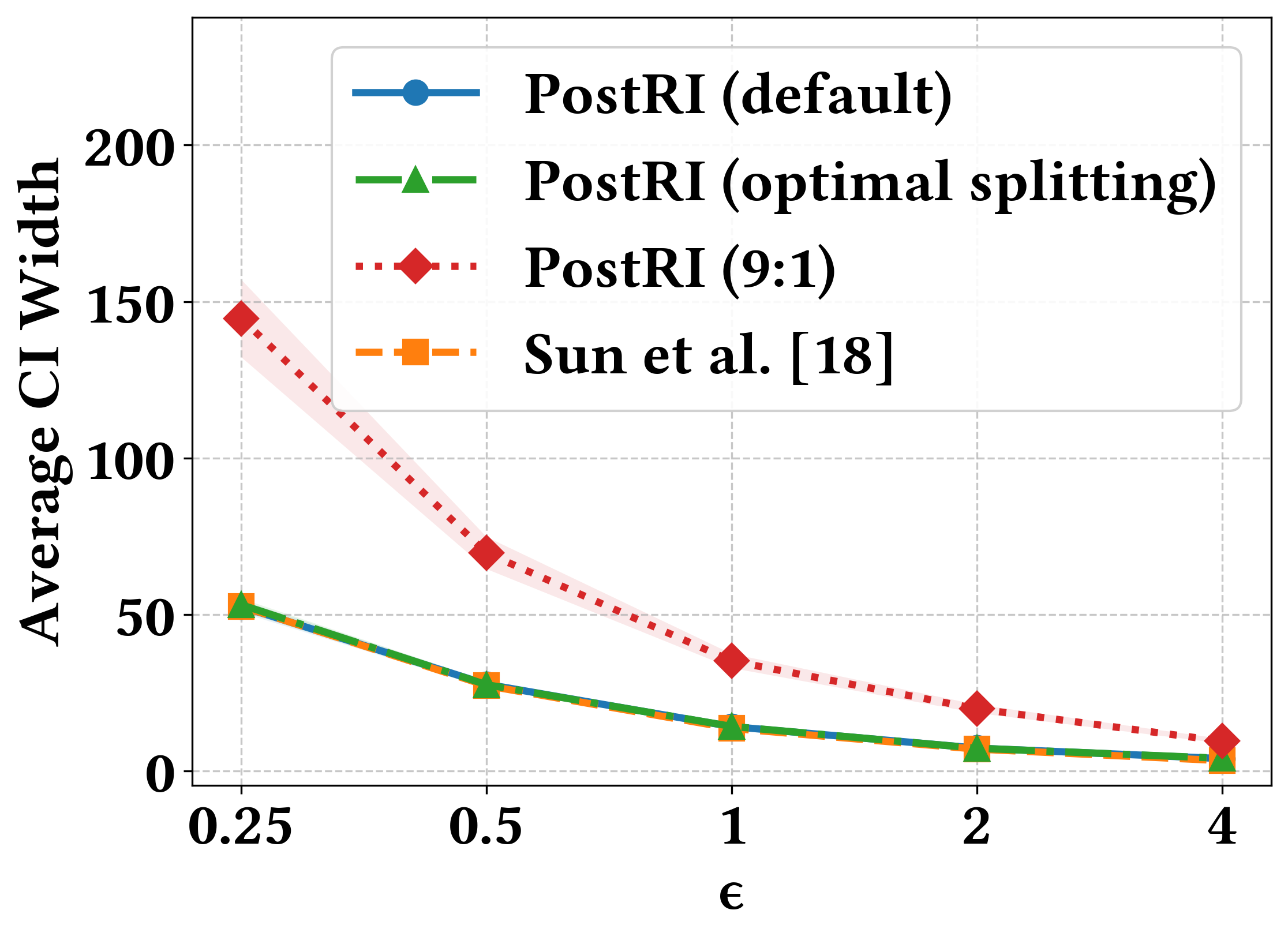}
  \end{subfigure}
  \caption{Average median error and RI length vs. varying privacy budget on the Banking dataset}
  \label{fig:varying_eps_banking}

  \begin{subfigure}[t]{0.48\columnwidth}
    \centering
    \includegraphics[width=\linewidth]{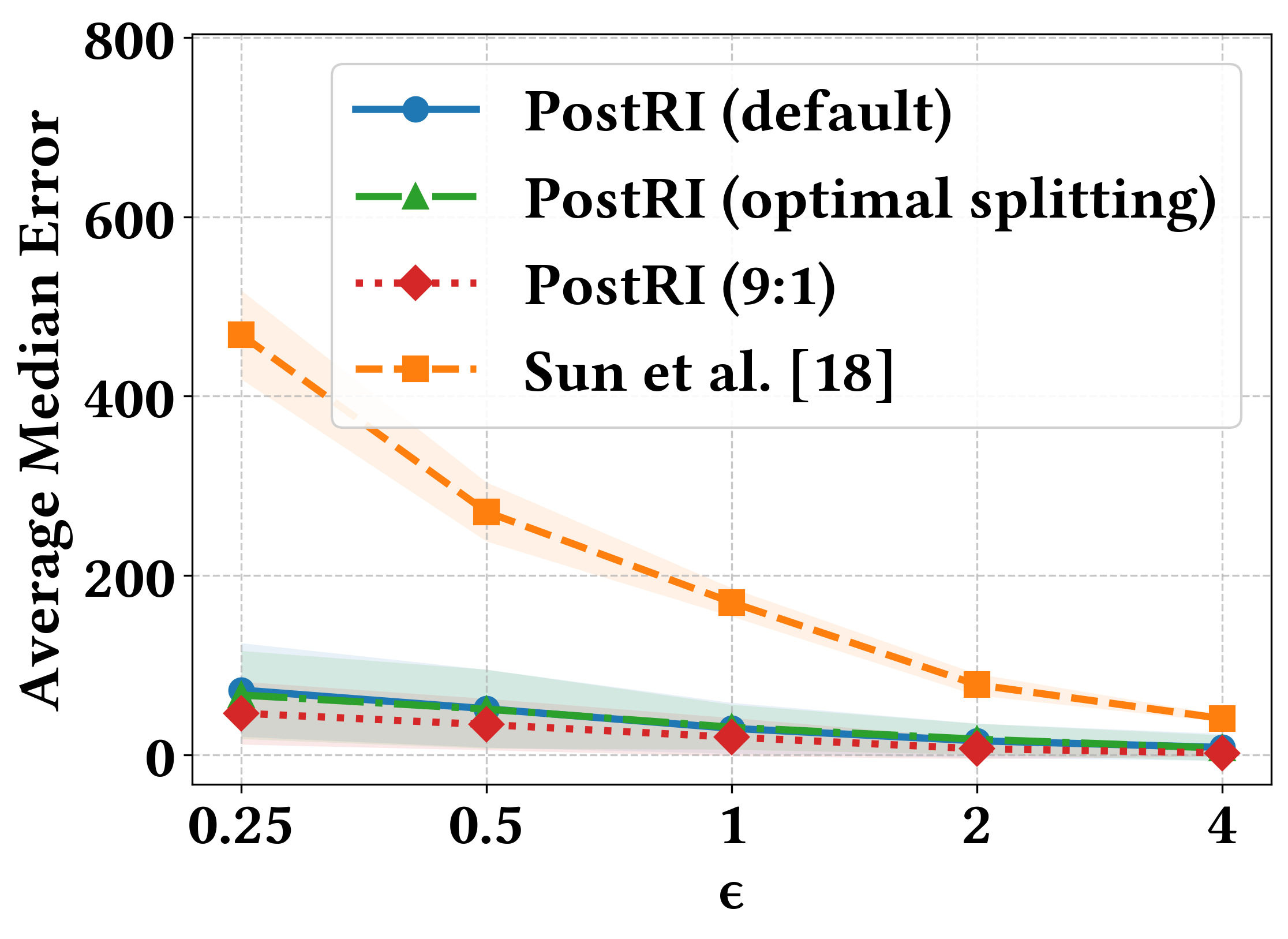}
  \end{subfigure}\hfill
  \begin{subfigure}[t]{0.48\columnwidth}
    \centering
    \includegraphics[width=\linewidth]{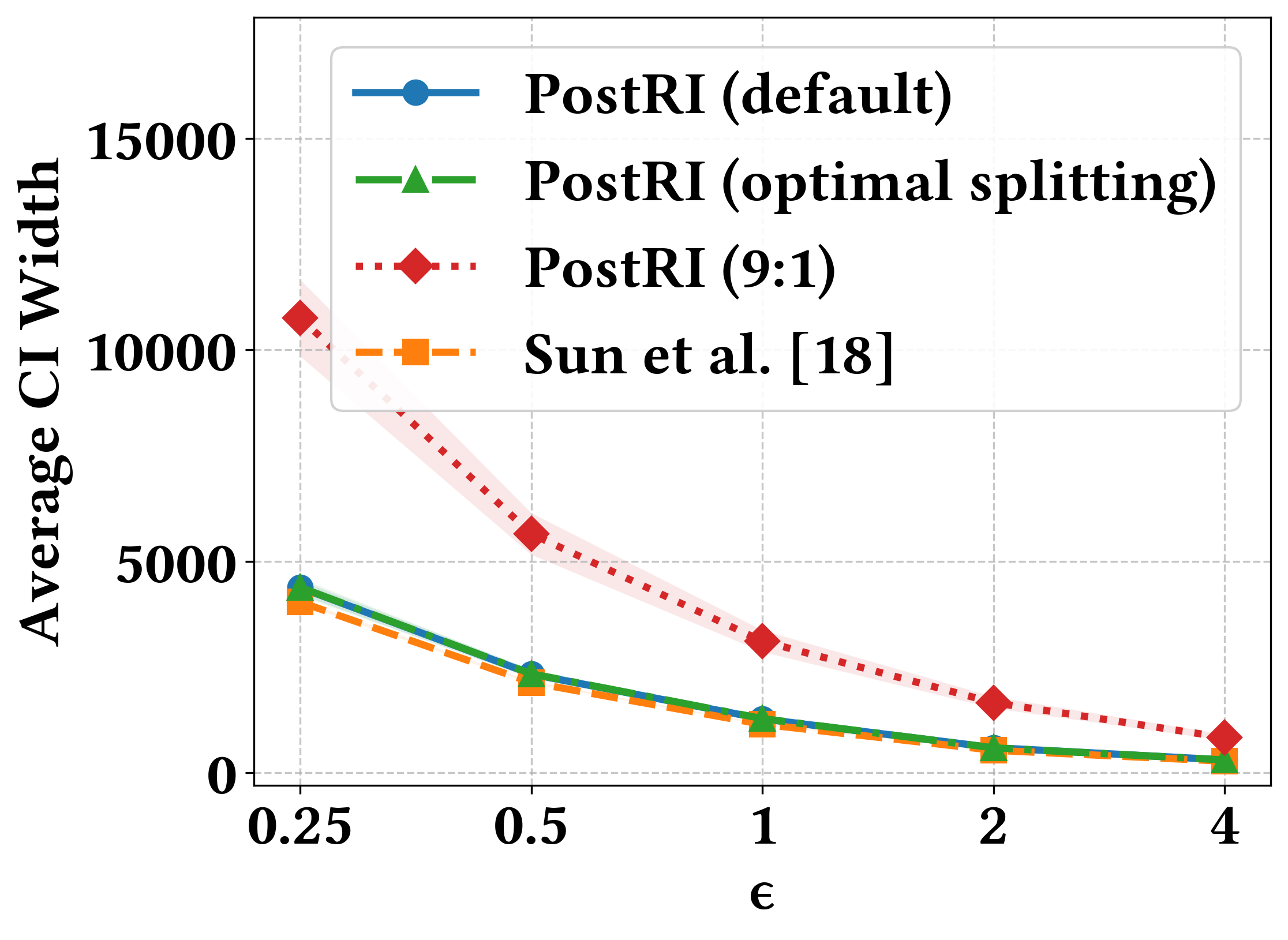}
  \end{subfigure}
  \caption{Average median error and RI length vs. varying privacy budget on the Adult dataset}
  \label{fig:varying_eps_adult}

  \begin{subfigure}[t]{0.48\columnwidth}
    \centering
    \includegraphics[width=\linewidth]{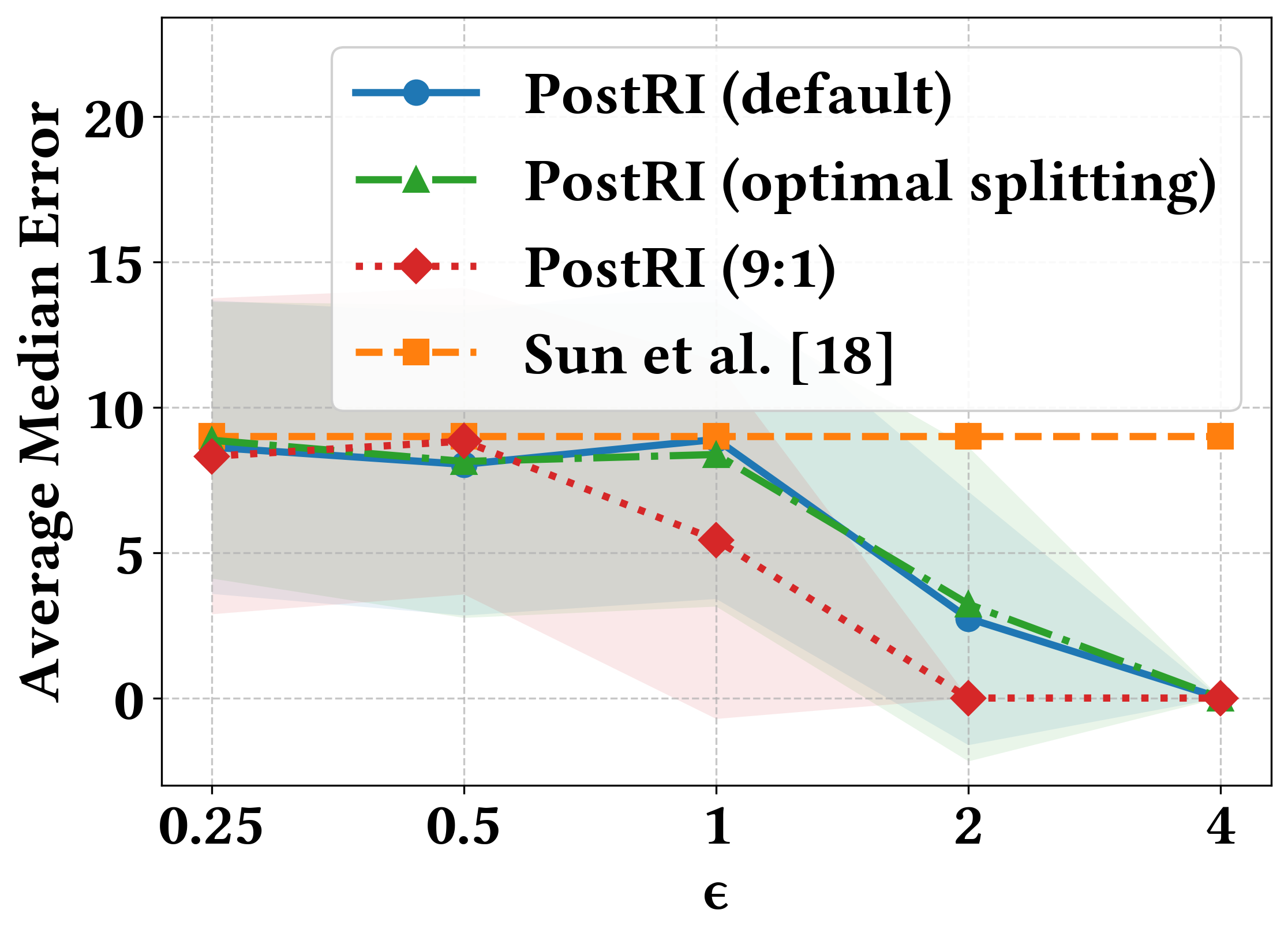}
  \end{subfigure}\hfill
  \begin{subfigure}[t]{0.48\columnwidth}
    \centering
    \includegraphics[width=\linewidth]{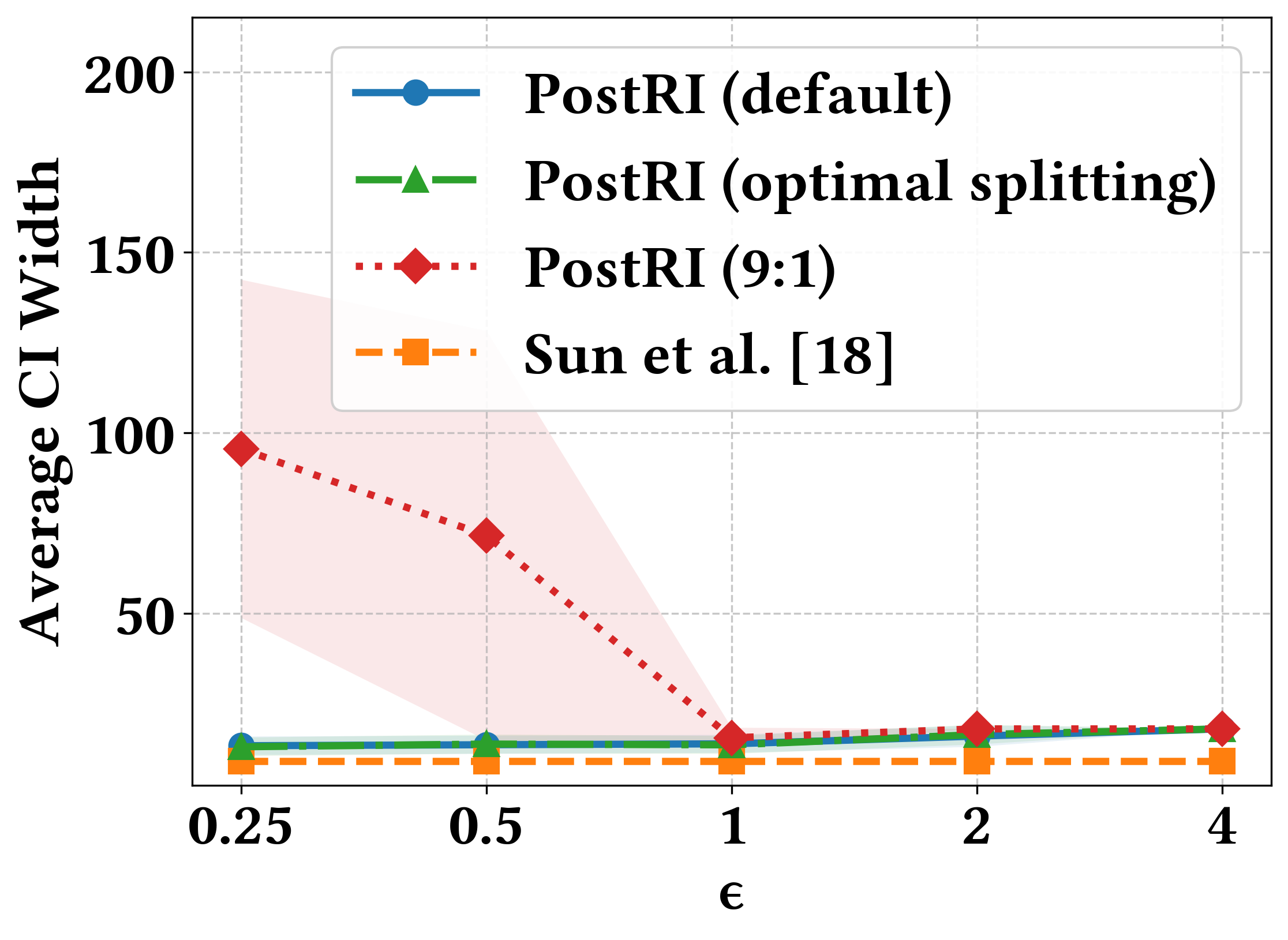}
  \end{subfigure}

  \caption{Average median error and RI length vs. varying privacy budget on the Airplane dataset}
  \label{fig:varying_eps_airplane}
\end{figure}




\section{Concluding Remarks}
In this work, we improve the state-of-the-art in computing randomization intervals for private median queries.
Our solution $\oursystem$ uses a novel utility function to enable a high utility median estimation along with a narrow randomization interval.
In the future, we look to extend this idea to other statistics.  

\section*{Acknowledgments}

This work was supported by NSERC through a Snowflake research fund, a Discovery Grant, and the Canada CIFAR AI Chairs program.

\bibliographystyle{ACM-Reference-Format}
\bibliography{tex/ref.bib}

@article{covington2021unbiased,
  title={Unbiased statistical estimation and valid confidence intervals under differential privacy},
  author={Covington, Christian and He, Xi and Honaker, James and Kamath, Gautam},
  journal={Statistica Sinica},
  year={to appear}
}

@article{dptextbook14,
 author = {Dwork, Cynthia and Roth, Aaron},
 title = {The Algorithmic Foundations of Differential Privacy},
 journal = {Found. Trends Theor. Comput. Sci.},
 year = {2014},
  publisher = {Now Publishers Inc.}
 }

@INPROCEEDINGS{Dwork06differentialprivacy,
    author = {Cynthia Dwork},
    title = {Differential privacy},
    booktitle = {IN ICALP},
    year = {2006},
    publisher = {Springer}
}

@article{nanayakkara2022visualizing,
  title={Visualizing privacy-utility trade-offs in differentially private data releases},
  author={Nanayakkara, Priyanka and Bater, Johes and He, Xi and Hullman, Jessica and Rogers, Jennie},
  journal={arXiv preprint arXiv:2201.05964},
  year={2022}
}

@inproceedings{hay2016principled,
  title={Principled evaluation of differentially private algorithms using dpbench},
  author={Hay, Michael and Machanavajjhala, Ashwin and Miklau, Gerome and Chen, Yan and Zhang, Dan},
  booktitle={Proceedings of the 2016 International Conference on Management of Data},
  pages={139--154},
  year={2016}
}

@inproceedings{xia2021dpgraph,
  title={Dpgraph: A benchmark platform for differentially private graph analysis},
  author={Xia, Siyuan and Chang, Beizhen and Knopf, Karl and He, Yihan and Tao, Yuchao and He, Xi},
  booktitle={Proceedings of the 2021 International Conference on Management of Data},
  pages={2808--2812},
  year={2021}
}

@article{liu2021catch,
  title={Catch a blowfish alive: a demonstration of policy-aware differential privacy for interactive data exploration},
  author={Liu, Jiaxiang and Knopf, Karl and Tan, Yiqing and Ding, Bolin and He, Xi},
  journal={Proceedings of the VLDB Endowment},
  volume={14},
  number={12},
  pages={2859--2862},
  year={2021},
  publisher={VLDB Endowment}
}

@article{sun2023confidence,
  title={Confidence Intervals for Private Query Processing},
  author={Sun, Dajun and Dong, Wei and Yi, Ke},
  journal={Proceedings of the VLDB Endowment},
  volume={17},
  number={3},
  pages={373--385},
  year={2023},
  publisher={VLDB Endowment}
}

@inproceedings{dwork2009differential,
  title={Differential privacy and robust statistics},
  author={Dwork, Cynthia and Lei, Jing},
  booktitle={Proceedings of the forty-first annual ACM symposium on Theory of computing},
  pages={371--380},
  year={2009}
}

@article{drechsler2022nonparametric,
  title={Nonparametric differentially private confidence intervals for the median},
  author={Drechsler, J{\"o}rg and Globus-Harris, Ira and Mcmillan, Audra and Sarathy, Jayshree and Smith, Adam},
  journal={Journal of Survey Statistics and Methodology},
  volume={10},
  number={3},
  pages={804--829},
  year={2022},
  publisher={Oxford University Press}
}

@inproceedings{cohen2024lower,
  title={Lower bounds for differential privacy under continual observation and online threshold queries},
  author={Cohen, Edith and Lyu, Xin and Nelson, Jelani and Sarl{\'o}s, Tam{\'a}s and Stemmer, Uri},
  booktitle={The Thirty Seventh Annual Conference on Learning Theory},
  pages={1200--1222},
  year={2024},
  organization={PMLR}
}

@article{ligett2025differentially,
  title={DIFFERENTIALLY PRIVATE NON-PARAMETRIC CONFIDENCE INTERVALS},
  author={Ligett, Katrina and Shenfeld, Moshe and Shoham, Tomer and Velner-Harris, Noa},
  journal={Journal of Privacy and Confidentiality},
  year={2025},
  
}

@article{chadha2024resampling,
  title={Resampling methods for private statistical inference},
  author={Chadha, Karan and Duchi, John and Kuditipudi, Rohith},
  journal={arXiv preprint arXiv:2402.07131},
  year={2024}
}

@article{lyu_2017,
  title = {Understanding the Sparse Vector Technique for Differential Privacy},
  author = {Lyu, Min and Su, Dong and Li, Ninghui},
  date = {2017-02},
  journaltitle = {Proc. VLDB Endow.},
  volume = {10},
  number = {6},
  pages = {637--648},
  publisher = {VLDB Endowment},
  issn = {2150-8097},
  doi = {10.14778/3055330.3055331},
  url = {https://doi.org/10.14778/3055330.3055331},
  issue_date = {February 2017},
  pagetotal = {12},
}

@inproceedings{mcsherry_2007,
  title = {Mechanism Design via Differential Privacy},
  booktitle = {Proceedings of the 48th Annual {{IEEE}} Symposium on Foundations of Computer Science},
  author = {McSherry, Frank and Talwar, Kunal},
  date = {2007},
  series = {{{FOCS}} '07},
  pages = {94--103},
  publisher = {IEEE Computer Society},
  location = {USA},
  doi = {10.1109/FOCS.2007.41},
  url = {https://doi.org/10.1109/FOCS.2007.41},
  isbn = {0-7695-3010-9},
}

@book{li_DP_book,
  title={Differential privacy: From theory to practice},
  author={Li, Ninghui and Lyu, Min and Su, Dong and Yang, Weining},
  year={2017},
  publisher={Springer}
}

@inproceedings{joint_exponential,
  title = {A Joint Exponential Mechanism for Differentially Private Top-{{k}}},
  booktitle = {Proceedings of the 39th International Conference on Machine Learning},
  author = {Gillenwater, Jennifer and Joseph, Matthew and Munoz, Andres and Diaz, Monica Ribero},
  editor = {Chaudhuri, Kamalika and Jegelka, Stefanie and Song, Le and Szepesvari, Csaba and Niu, Gang and Sabato, Sivan},
  year = {2022-07-17/2022-07-23},
  series = {Proceedings of Machine Learning Research},
  volume = {162},
  pages = {7570--7582},
  publisher = {PMLR},
  url = {https://proceedings.mlr.press/v162/gillenwater22a.html},
}

@article{du2020differentially,
  title={Differentially private confidence intervals},
  author={Du, Wenxin and Foot, Canyon and Moniot, Monica and Bray, Andrew and Groce, Adam},
  journal={arXiv preprint arXiv:2001.02285},
  year={2020}
}

@article{panavas2024illuminating,
  title={Illuminating the Landscape of Differential Privacy: An Interview Study on the Use of Visualization in Real-World Deployments},
  author={Panavas, Liudas and Sarker, Amit and Di Bartolomeo, Sara and Sarvghad, Ali and Dunne, Cody and Mahyar, Narges},
  journal={IEEE Transactions on Visualization and Computer Graphics},
  year={2024},
  publisher={IEEE}
}

@inproceedings{DinurN03,
  author       = {Irit Dinur and
                  Kobbi Nissim},
  editor       = {Frank Neven and
                  Catriel Beeri and
                  Tova Milo},
  title        = {Revealing information while preserving privacy},
  booktitle    = {Proceedings of the Twenty-Second {ACM} {SIGACT-SIGMOD-SIGART} Symposium
                  on Principles of Database Systems, June 9-12, 2003, San Diego, CA,
                  {USA}},
  pages        = {202--210},
  publisher    = {{ACM}},
  year         = {2003},
  url          = {https://doi.org/10.1145/773153.773173},
  doi          = {10.1145/773153.773173},
  bibsource    = {dblp computer science bibliography, https://dblp.org}
}


\appendix

\section{Proofs}
\subsection{Proof of Theorem~\ref{thm:RIPrivacy}}\label{sec:proof_RI_Privacy}
We first prove the following lemma.
\begin{lemma}\label{lem:sens_f}
    $\Helper$ has sensitivity $1$. Formally:
    \begin{equation}
    \max_{\substack{\Dataset, \Dataset'\in \mc{D}\\ \RIStep \in \{\StepSize, 2\StepSize, \dots, \lfloor \frac{\StepSize \cdot \DomainSize}{\StepSize} \rfloor \}}}\left|\Helper_\RIStep(\Dataset) - \Helper_\RIStep(\Dataset')\right|\leq 1
    \end{equation}
\end{lemma}
\begin{proof}
Recall $\HelperFunction{\RIStep} = \min(| \Rank{\NoisyMedian+\RIStep} - \Rank{\NoisyMedian} |,  | \Rank{\NoisyMedian} - \Rank{\NoisyMedian-\RIStep} |)$.
w.l.o.g assume that $\Rank{\NoisyMedian - \RIStep} \leq \Rank{\NoisyMedian} \leq \Rank{\NoisyMedian + \RIStep}$.
Then, the term $|\Rank{\NoisyMedian - \RIStep} - \Rank{\NoisyMedian}|$ is equivalent to the number of data points that fall in the interval $(\NoisyMedian - \RIStep, \NoisyMedian]$. Replacing a data point can, in the worst case, move a point in or out of this interval, implying a sensitivity of one.
A symmetric argument holds for $|\Rank{\NoisyMedian + \RIStep} - \Rank{\NoisyMedian}|$.
Finally, we consider the min. The min either remains unchanged or the min swaps under the replacement of a record. 
In either case, the output can change by at most one, as both terms can change by at most one.
\end{proof}

Now we can prove the Theorem:
\RIPrivacy*
\begin{proof}
   The sensitivity of $\RIUtilSymbol$ is equivalent to the sensitivity of $\Helper$ since all other terms are constant over neighbouring datasets.
   Then, by the privacy properties of the exponential mechanism, the result follows.
\end{proof}

\subsection{Proof of Theorem~\ref{thm:RICorrect}}\label{sec:correctProof}

\RICorrect*
\begin{proof}
With probability $1-\beta_1$ since $o$ is the output of the exponential mechanism, we have:
\begin{equation}\label{eqn:textbook_utility}
    \MedianUtility{\NoisyMedian} \geq \MedianUtility{\MaxUtility}-\EMBound_1
\end{equation}
subbing in the definition of the median utility function (\ref{eq:median_utility}), gives the following
\begin{equation}
    |\Rank{\NoisyMedian} - \DataSize/2| \leq \EMBound_1
\end{equation}
which implies
\begin{equation}\label{eq:em_statment}
\Rank{\NoisyMedian}-\EMBound_1 \leq \DataSize/2 \leq \Rank{\NoisyMedian} + \EMBound_1.
\end{equation}

For a given output $\hat{\RIStep}$ of the second exponential mechanism, with probability $1-\beta_2$, we have the following statement 
\begin{equation}\label{eqn:textbook_utility_2}
    \RIUtility{\hat{\RIStep}} \geq \RIUtility{\MaxRIUtility} -\EMBound_2 -\StepSize\cdot\Lipshitz
\end{equation}
which gives
\begin{equation}\label{eq:utility_bound_2}
    |\HelperFunction{\hat{\RIStep}} - \EMBound_1 - \EMBound_2 -\StepSize\cdot\Lipshitz| \leq \EMBound_2 + \StepSize\cdot\Lipshitz.
\end{equation}

W.l.o.g assume that $\Rank{\NoisyMedian - \hat{\RIStep}}\leq \Rank{\NoisyMedian} \leq \Rank{\NoisyMedian + \hat{\RIStep}}$. 
We consider two cases $\Rank{\NoisyMedian} \geq \DataSize/2$ and $\Rank{\NoisyMedian} < \DataSize/2$.
In the first case, if $\Rank{\NoisyMedian} \geq \DataSize/2$ then by assumption we have $\Rank{\NoisyMedian + \hat{\RIStep}} \geq \Rank{\NoisyMedian} \geq \DataSize/2$. We must show $\Rank{\NoisyMedian - \hat{\RIStep}}<\DataSize/2$.
Expanding Eqn~\ref{eq:utility_bound_2} we get the following
\begin{equation}
    -\EMBound_2 - \StepSize\cdot\Lipshitz \leq \HelperFunction{\hat{\RIStep}} - \EMBound_1 - \EMBound_2 -\StepSize\cdot\Lipshitz
\end{equation}
where the absolute value is omitted following our initial assumption. Then we get 
\begin{equation}
    -\EMBound_2 - \StepSize\cdot\Lipshitz \leq \Rank{\NoisyMedian} - \Rank{\NoisyMedian - \hat{\RIStep}} - \EMBound_1 - \EMBound_2 -\StepSize\cdot\Lipshitz.
\end{equation}
We note that this holds regardless of which term is the min in $\HelperFunction{\hat{\RIStep}}$.
Then rearranging and subbing in Eqn~\ref{eq:em_statment}, we get
\begin{equation}
    \Rank{\NoisyMedian - \hat{\RIStep}} \leq \Rank{\NoisyMedian} - \EMBound_1 \leq \DataSize/2.
\end{equation}

In the second case, if $\Rank{\NoisyMedian} < \DataSize/2$ then by assumption we have $\Rank{\NoisyMedian - \hat{\RIStep}} \leq \Rank{\NoisyMedian} < \DataSize/2$. We must show $\Rank{\NoisyMedian + \hat{\RIStep}}>\DataSize/2$.
Recall that expanding Eqn~\ref{eq:utility_bound_2} we get the following
\begin{equation}
    -\EMBound_2 - \StepSize\cdot\Lipshitz \leq \HelperFunction{\hat{\RIStep}} - \EMBound_1 - \EMBound_2 -\StepSize\cdot\Lipshitz.
\end{equation}
Then we get 
\begin{equation}
    -\EMBound_2 - \StepSize\cdot\Lipshitz \leq \Rank{\NoisyMedian + \hat{\RIStep}} - \Rank{\NoisyMedian} - \EMBound_1 - \EMBound_2 -\StepSize\cdot\Lipshitz.
\end{equation}
We similarly note that this holds regardless of which term is the min in $\HelperFunction{\hat{\RIStep}}$.
Finally, rearranging and subbing in Eqn~\ref{eq:em_statment}, we get
\begin{equation}
    \Rank{\NoisyMedian + \hat{\RIStep}} \geq \Rank{\NoisyMedian} + \EMBound_1 \geq \DataSize/2.
\end{equation}
Taking the union bound over the probability of Eqn~\ref{eqn:textbook_utility} and Eqn~\ref{eqn:textbook_utility_2}, the result follows.
\end{proof}

\subsection{Epsilon Splitting Analysis}\label{sec:epsilonsplitting}
We minimize the following error
\begin{equation}
\EMBound_1 + \EMBound_2 + \StepSize\cdot \Lipshitz
\end{equation}
using Lagrange multipliers with the condition that $\epsilon_1 +\epsilon_2 = \epsilon$.
Let
\begin{equation}
    L = \frac{2}{\epsilon_1} \log{\frac{\DomainSize}{\beta_1}} + \frac{2}{\epsilon_2} \log{\frac{\DomainSize}{\StepSize\beta_2}} + \StepSize\cdot \ell + \lambda(\epsilon_1 + \epsilon_2 - \epsilon).
\end{equation}
Then, differentiating, we get 
\begin{equation}
    \frac{\partial L}{\partial \epsilon_1} = -\frac{2}{\epsilon_1^2} \log{\frac{\DomainSize}{\beta_1}} + \lambda
\end{equation}
\begin{equation}
    \frac{\partial L}{\partial \epsilon_2} = -\frac{2}{\epsilon_2^2} \log{\frac{\DomainSize}{\StepSize\beta_2}} + \lambda.
\end{equation}
Setting these equal and rearranging gives us the following:
\begin{equation}
    \epsilon_1^2 = \epsilon_2^2 \cdot \frac{\log{\frac{\DomainSize}{\beta_1}}}{\log{\frac{\DomainSize}{\StepSize\beta_2}}}
\end{equation}
Taking the square root gives us the final result:
\begin{equation}
    \epsilon_1 = \epsilon_2\sqrt{\frac{\log{\frac{\DomainSize}{\beta_1}}}{\log{\frac{\DomainSize}{\StepSize\beta_2}}}}
\end{equation}

\subsection{Optimal Step Size Analysis} \label{sec:optimalstepsize}
Let us assume that the RI width of our algorithm is $\EMBound_1 +\EMBound_2 + \StepSize\cdot \Lipshitz$. We assume all other variables are constant.
Assume we always conduct queries for $b \in \{\StepSize, 2\StepSize, \dots \lfloor \frac{\StepSize\DomainSize}{\StepSize}\rfloor $.
Then, ignoring constants (the median error), the width of the RI is
\begin{equation}
    \EMBound_2 + \StepSize\cdot \ell = \frac{2\RISensitiv}{\epsilon_2} \log{\frac{\DomainSize}{\StepSize\beta_2}} + \StepSize\cdot \ell
\end{equation}
The first derivative of this w.r.t. $\StepSize$ is
\begin{equation}
    -\frac{2\RISensitiv}{\epsilon_2\StepSize} + \ell
\end{equation}
Solving for $\StepSize$ gives the optimal way to set this parameter
\begin{equation}
    \StepSize = \frac{2\RISensitiv}{\epsilon_2 \ell}
\end{equation}
where $\ell$ is the Lipschitz bound on the RI utility function.

\section{Utility Analysis}\label{sec:utility_analysis}
To compare the utility to that of Sun et al.~\cite{sun2023confidence}, we first state the utility bound from their paper. Namely, for any $y\in[L(D), R(D)]$ (in the RI), they show that 
\begin{equation}
    \MedianUtility{y} \geq \MedianUtility{\MaxUtility}- \frac{17\MedianSensitiv}{\epsilon} \log{\frac{2\DomainSize}{\beta}} - 2\Lipshitz
\end{equation}

To roughly compare the utilities of the work, we first define a constant term of factors common to both approaches.
\begin{equation}
    \eta = \frac{2\MedianSensitiv}{\epsilon} \log{\frac{2\DomainSize}{\beta}}
\end{equation}

Using $\eta$, we make the following observations.
\begin{claim}\label{claim:equalRI}
     Assuming the best case where the exponential mechanisms return the candidates with optimal utility, the width of the RI for both approaches will be approximately $8\eta$
\end{claim}
\begin{proof}[Proof Sketch]   
    The utility function of Sun et al. is 
    $\MedianUtility{y} - \MedianUtility{\MaxUtility} - s - \ell$, where $\MedianUtility{\MaxUtility}=0$, $s\approx4\eta$, and we will ignore $\ell=1$. This means the distance from the median to the lower bound is approximately $4\eta$. The upper bound is similar, resulting in a total width of $8\eta$.
    
    Our utility function is $|\HelperFunction{\hat{\RIStep}} - \EMBound_1 - \EMBound_2 -\StepSize\cdot\Lipshitz|$
    where $\EMBound_1 \approx 2\eta$, $\EMBound_2 \approx 2\eta$, and we ignore the $\ell$ term. This gives an optimal helper function $\Helper$ of $4\eta$ (the worst-case distance to either the upper or lower bound), which implies a total width of $8\eta$.
\end{proof}

\begin{claim}\label{claim:worstWidth}
    Assuming the worst case where the exponential mechanisms return the candidates with the worst possible utility defined by Theorem~\ref {thm:median_utility}. Then Sun et al. have an RI width of $16\eta$ and \oursystem has an RI width of $12\eta$.
\end{claim}
\begin{proof}[Proof Sketch]  
    Building on Claim~\ref{claim:equalRI}, we can assume the error of each application of the exponential mechanism in Sun et al. is $4\eta$ by applying the utility bound of the exponential mechanism with an extra factor of 2 from splitting epsilon over the lower and upper bound and another 2 from Lemma 3.3 of Sun et al.~\cite{sun2023confidence}.
    This means the lower bound estimates the distance of $4\eta$ from the median with at worst a $4\eta$ error, resulting in a total distance of $8\eta$. Similarly, for the upper bound, giving the total width of $16\eta$.
    
    Our work builds the RI using the estimated median as the center. Then we estimate a width of $4\eta$ (as discussed in Claim~\ref{claim:equalRI}) with error at most $\EMBound_2 \approx 2\eta$. This leads to a total distance of $6\eta$  per side and an overall distance of $12\eta$.
\end{proof}

\begin{claim}
    Assuming the worst case where the exponential mechanisms return the candidates with the worst possible utility defined by Theorem~\ref {thm:median_utility}. Then Sun et al. have a median error of $4\eta$ and \oursystem has a median error of $2\eta$.
\end{claim}
\begin{proof}[Proof Sketch]  
    Building on Claim~\ref{claim:worstWidth}, if we assume the lower bound is the worst possible distance of $8\eta$ and in the worst case, the upper bound is the smallest possible distance of 0 from the true median (estimating $4\eta$ distance with error $4\eta$ that cancel each other out). Then the average of the lower and upper bounds is $4\eta$ away from the true median.
    
    Our work simply estimates the median with error at most $\EMBound_1 \approx 2\eta$. 
\end{proof}


\section{Related Work}\label{sec:related_work}

There has been a significant amount of work done on expressing the privacy-utility trade-off present in any DP mechanism~\cite{Dwork06differentialprivacy} to a non-expert.  An approach can be to provide a set of benchmark utility scores for common privacy parameter settings and datasets~\cite{xia2021dpgraph,hay2016principled}. A more usable approach is to provide scores as an interactive visualization tool~\cite{liu2021catch,xia2021dpgraph,nanayakkara2022visualizing,panavas2024illuminating}, so that users can explore possible privacy settings.

Our work looks at expressing the privacy utility trade-off using randomization intervals. We note that although randomization intervals provide a bound on the DP noise with high confidence, they do not provide a bound on the variance from the data (sample variance). Most prior work on DP confidence intervals focuses on bounding this error without considering the error from the noise itself ~\cite{dwork2009differential,du2020differentially,covington2021unbiased,cohen2024lower,ligett2025differentially,chadha2024resampling}.

In our work, we treat finding bounds for a confidence intervals as a maximizing a utility function problem. This then leads us to using an exponential mechanism based design. Sun et al.\cite{sun2023confidence} also take this approach. They provide confidence intervals for the median and other statistics using both exponential mechanism and svt based approaches. Their approach first estimates valid upper and lower bounds for the statistic, before taking the average of the values. 

Drechsler et al.\cite{drechsler2022nonparametric} consider solutions that estimate confidence intervals using both the exponential mechanism and post-processing an estimated DP CDF. Their solutions aim to to account for both sources of randomness (sample error and privacy noise). However, in both of their approaches, they do not directly estimate the median itself.

\section{Implementation Details}\label{sec:implementation}
Our implementation is based on the code of Sun et al.~\cite{sun2023confidence}:  \url{https://github.com/PrivateCI/DP_CI}. We note that running their code unmodified did not reproduce the results Sun et al.'s paper exactly. We modified the domain size in their code to $\DomainSize=10^8$ to be consistent across both approaches. We also note that the parameter $s$ in Sun et al.'s work is set to $\frac{9\MedianSensitiv}{\epsilon} \log{\frac{2\DomainSize}{\beta}}$ in the paper, but $\frac{8\MedianSensitiv}{\epsilon} \log{\frac{2\DomainSize}{\beta}}$ in the code. We follow their code as it gives a lower error and matches the settings of our work. We also fixed a couple of off by one errors in Sun et al.'s code. Our code can be found at \url{https://github.com/Timliuw/Randomization-Intervals}.

\end{document}